\documentclass[psamsfonts, reqno]{amsart}

\usepackage{amssymb,amsfonts}
\usepackage[all,arc]{xy}
\usepackage{enumerate}
\usepackage{mathrsfs}
\usepackage{mathtools}
\usepackage{amsmath}
\usepackage{lipsum}
\usepackage{changepage}
\usepackage{ upgreek }
\usepackage{bm}
\usepackage{tabu}
\usepackage{chngcntr}
\usepackage{graphicx}

\newtheorem{thm}{Theorem}[section]

\newtheorem{lem}[thm]{Lemma}

\theoremstyle{definition}

\theoremstyle{remark}


\title{Optimal Estimating Equation for Logistic Regression with Linked Data}
 
\author{%
J\MakeLowercase{enkin} T\MakeLowercase{sui}\textsuperscript{1}, A\MakeLowercase{bel} D\MakeLowercase{asylva}, K\MakeLowercase{enneth} C\MakeLowercase{hu}\\
\textbf{}{\\}
\textsuperscript{1}D\MakeLowercase{epartment} o\MakeLowercase{f} C\MakeLowercase{omputer} a\MakeLowercase{nd} M\MakeLowercase{athematical} S\MakeLowercase{ciences}, U\MakeLowercase{niversity} o\MakeLowercase{f} T\MakeLowercase{oronto} S\MakeLowercase{carborough}, T\MakeLowercase{oronto}, O\MakeLowercase{ntario}, M1C 1A4, C\MakeLowercase{anada}
}

\date{JULY 22, 2017}

\begin{document}

\begin{abstract}

We propose an optimal estimating equation for logistic regression with linked data while accounting for false positives.
 It builds on a previous solution but estimates the regression coefficients with a smaller large sample variance.\\


\end{abstract}

\maketitle

\tableofcontents

\section{Introduction}

{\it Record Linkage}, also known as data linkage, is the act of bringing together records from two files, say file $X$ and file $Y$, that relate to the same individual or entity.
 A record-pair is {\it matched} when the two records belong to the same individual.
 Otherwise it is {\it unmatched}.
 A record is {\it unlinked} if it has no link to any outgoing record in the other file.
 A link between two unmatched records is called a {\it false positive}.
 The absence of a link between two matched records is called a {\it false negative}.
 Note that it is common to have false positives and unlinked records for a variety of reasons such as: human entry errors,
 change of formats overtime, nonunique values, missing values, etc.

 There are many record-linkage methods.
 In the {\it probabilistic} method two records are linked according to the probability that they are matched given their comparison outcomes \cite{fellegi_sunter_1969}.
 This is particularly useful when the linkage information is limited to {\it quasi-identifiers} (e.g. the first name, demographic variables, the postal code, etc.) instead of unique identifiers.

 This paper looks at logistic regression with linked data, while accounting for false positives.
 It improves a previous estimator by Chipperfield et al. \cite{chipperfield_bishop_campbell_2011}.


\section{Background}

Linkage errors include false positives and false negatives.
There are obvious parallels with traditional surveys.
Indeed false positives are analogous to measurement errors, while false negatives and unlinked records are analogous to item nonresponse.
All these errors are potential sources of bias in regressions with linked data.
The problem of regression under linkage errors has been previously discussed by Neter et al. \cite{ neter_maynes_ramanathan_1965},  Scheuren and Winkler \cite{ scheuren_winkler_1993}, and Krewski et al. \cite{ krewski_dewanji_2001}.
Various solutions have been described depending on the available linkage information, the linkage scenario and the regression model.
In a primary analysis, all details of the linkage project are available to the analyst.
For a probabilistic linkage, this information include the linkage weight of each record-pair.
In this case, the study may also include all the selected record-pairs, regardless of whether they have been linked, as suggested by Scheuren and Winkler \cite{ scheuren_winkler_1993}.
At Statistics Canada, such studies are conducted by analysts within the agency.
In a secondary analysis \cite{ kim_chambers_2016}, the analyst has access to limited information about the linkage project, e.g. the overall rates of false positives and false negatives.
A secondary analyst may be an academic researcher at a research data center.
As mentioned before, the linkage scenario is another important factor.
It describes whether each record in each file is matched to exactly one record (one-to-one and onto), at most one record (one-to-one), or many records (one-to-many or many-to-many).
Finally, the actual solution also depends on the regression model.
Previous studies cover linear regression, logistic regression, contingency tables, mortality studies, as well as capture-recapture models. 
Scheuren and Winkler \cite{ scheuren_winkler_1993, scheuren_winkler_1997} have developed a bias-correction solution for linear regression by a primary analyst, including features to deal with outliers.
The solution considers a probabilistic linkage and exploits the linkage weights that are estimated by a model.
Although the resulting estimator performs well in simulations , it is biased.
Lahiri and Larsen \cite{lahiri_larsen_2005} have improved the solution by Scheuren and Winkler for a one-to-one-and-onto linkage, where the matched record-pairs are fully described by a permutation matrix, that is hereafter called match matrix.
They have proposed a least-squares estimator, which requires the expected match matrix and is unbiased provided the match matrix and the vector of responses are conditionally independent given the covariates.
For a probabilistic linkage, Lahiri and Larsen have also proposed the estimation of the expected match matrix with mixture models of the record-pairs.
Chipperfield et al. \cite{chipperfield_bishop_campbell_2011} have considered the situation of a primary analyst who performs a logistic regression or the analysis of a contingency table, using the linked pairs and a sample of such pairs, which are each known to be matched or unmatched through clerical reviews.
Clerical-reviews are based on visual inspections of the sampled pairs  by qualified personnel to determine if they are matched.
The proposed solution is inspired by the maximum-likelihood framework and incorporates separate features for false positives and unlinked records.
For the false positives, the adjustment is based on the estimated probability that a link is matched given the covariates and the observed response; the probability being estimated from the clerical sample.
As for unlinked records, the adjustment is based on a reweighting of the linked records.
Dasylva\cite{dasylva_2014} has designed a calibration solution for estimation with linked data, when the linkage is based on the probabilistic methodology.
This solution uses all the potential pairs, their linkage weights to predict the match status, as well as a clerical sample.
The estimated parameters are calibrated to control totals that are based on the predicted match status.
Krewski et al. \cite{krewski_dewanji_2001}, Mallick \cite{mallick_2005} and Wang and Donnan \cite{wang_donnan_2002} have discussed cohort mortality studies with  linked data.
Mallick \cite{mallick_2005} has proposed a bias-correction solution for primary analysts.
It uses a clerical-review sample to account for the bias due to linkage errors.
Wang and Donnan \cite{wang_donnan_2002} have proposed a solution to account for the unlinked records, when the links are missing at random.
A secondary analysis of linked data is a greater challenge than a primary analysis  because of the limited information about  the linkage.
Fortunately, many solutions also exist in the literature after the landmark paper by Chambers \cite{chambers_2009a}.
In that work, Chambers has described estimating equations for linear and logistic regression with linked data, including a Least Squares Estimator (LSE), a Best Linear Unbiased Estimators (BLUEs) and an Empirical BLUE for each model (linear or logistic).
The proposed estimators are consistent when the links are Incorrect at Random (IAR), i.e. when the linkage errors and the responses are conditionally independent given the responses.
 They require the rates of false positives and false negatives in strata of record-pairs called blocks, which are used to select a reasonably small subset of pairs from the Cartesian product of two large files.
The original solution by Chambers \cite{chambers_2009a} has been extended in many directions, including the linkage of a sample to a register \cite{kim_chambers_2012a}, finite population inference \cite{chambers_2009b}, and the probabilistic linkage of three files \cite{kim_chambers_2012b}, where a first file contains the responses, while the remaining files contain the covariates.
Other solutions for the analysis of linked data have been proposed including imputation solutions by Larsen \cite{larsen_1999} and Goldstein \cite{goldstein_2015}, and bayesian solutions by Tancredi and Liseo \cite{tancredi_liseo_2015}.
 Ding and Feinberg (1994) and Di Consiglio and Tuoto (2015) have discussed linkage-errors in the context of capture-recapture models for coverage studies.

The focus of this study is on logistic regression by a primary analyst, when there is also clerical-review sample, for example to measure the rates of linkage error.
In that regard, the methodology by Chipperfield et al. \cite{chipperfield_bishop_campbell_2011}  is of special interest to exploit all the available information.
 Besides this methodology dispenses with the requirement that the linkage be Incorrect at Random (IAR), unlike Chambers \cite{ chambers_2009a}, Chambers et al. \cite{chambers_2009b},
 and Kim and Chambers \cite{kim_chambers_2012a, kim_chambers_2012b, kim_chambers_2013, kim_chambers_2016}.
However, the original solution has been modified and improved using the quasi-likelihood framework \cite{heyde}.


\section{Notation and assumptions}\label{section: notations}

Following Chipperfield et al.\cite{chipperfield_bishop_campbell_2011}, consider two files X and Y, which record the characteristics of the same finite population of individuals.
 Each individual is characterized by a vector of covariates $\mathbf{X}$ and a binary response $Y$, which are {\it both random} and related by a logistic model.
 For convenience, assume that the first component of $\mathbf{X}$ is always 1.
 For each individual the covariates are recorded in file X (as the vector $\mathbf{X}$ for the corresponding record), while the response is separately recorded in file Y (as the variable $Y$ for the related record).
 Additionally, both files contain linkage variables that are based on the recording of identifiers or quasi-identifiers associated with an individual.
 In general the linkage variables are affected by errors.
 The covariates and response are also susceptible to recording errors.
However these other errors are hereafter ignored for simplicity and to be consistent with previous work.
 Each individual is recorded at most once in each file.
 The linkage variables are used to link the two files and produce links that are labeled from 1 to $N$.
 The linkage may also produce unlinked records that are hereafter assumed to occur {\it completely at random}.
 For the $i$-th link, let $\mathbf{X}_i$ denote the covariates and $Y_i^*$ the observed response.
Let $Y_i$ denote the actual response that is associated with the covariates $\mathbf{X}_i$ through the logistic model as follows.
\begin{equation}
E \left [ \left . Y_i \right | \mathbf{X}_i = \mathbf{x}_i \right ] =
P(Y_i =1|\mathbf{X}_i = \mathbf{x}_i) := \mu (\bm{\beta};\mathbf{x}_i):= \mu_i (\bm{\beta}) = \frac{ e^{\mathbf{x}_i^{\top}\bm{\beta}} }{1+ e^{\mathbf{x}_i^{\top}\bm{\beta}}}
\end{equation}
where $\bm{\beta} = \left [ \beta_0 \ldots \beta_{p-1} \right ]^{\top}$ is the vector of unknown regression coefficients.
In a matched pair, the observed response and the actual response are identical, i.e. $Y_i^* = Y_i$.
Let $D_i$ denote the indicator variable corresponsing to the match status of link $i$, with $D_i=1$ if it is matched, $D_i=0$ else.
In general this match status is unknown but my be determined through a clerical-review when there is enough data to support such a review.
These conditions are met when linking social data with names as in some census applications described by Chipperfield et al. \cite{chipperfield_bishop_campbell_2011}.
However clerical-reviews are expensive so that it is always desirable to minimize them, e.g., by drawing a reasonably small probabilistic sample of pairs or links for review.
For the problem at hand, suppose that a Bernoulli sample $s$ of links is drawn, where each link is selected with the probability $p$ independently of other links, with $R_i$
 denoting the indicator that link $i$ is selected for clerical-review, i.e. $R_i = I(i \in s)$.
 Although this sampling design departs from the original solution, it greatly simplifies our subsequent derivations.
 Furthermore, we believe that it does not substantively modify our conclusions.
 For simplicity and in keeping with \cite{chipperfield_bishop_campbell_2011}, the clerical-reviews are assumed error-free, such that for each reviewed link $i$, the clerical decision perfectly coincides with the match status $D_i$.
Let $O_i$ denote the observed data for link $i$.
 For a sampled link $i$, $O_i$ is comprised of the covariates $\mathbf{X}_i$, the observed response $Y_i^*$, the actual response $Y_i$, the sample selection indicator $R_i$, and the match status (or clerical-decision) $D_i$.
 For links outside the clerical sample, $O_i$ is limited to the covariates $\mathbf{X}_i$, the observed response $Y_i^*$, and the sample selection indicator $R_i$.
Having a clerical-sample provides a basis for estimating the conditional probability that a link $i$ is matched given the covariates and observed response, i.e. $P \left (D_i \left | {\bf X}_i, Y_i^* \right . \right )$, without making the assumption that the links are
 Incorrect at Random (IAR), unlike Chambers \cite{ chambers_2009a}, Chambers et al. \cite{chambers_2009b}, and Kim and Chambers \cite{kim_chambers_2012a, kim_chambers_2012b, kim_chambers_2013, kim_chambers_2016}.
 This conditional match probability is crucial to the adjustment mechanism for false positives in the methodology described by Chipperfield et al. \cite{chipperfield_bishop_campbell_2011}.
 Note that the conditional match probability determines any expectation of the form $E \left [ f({\bf X}_i, Y_i^*) P \left (D_i \left | {\bf X}_i, Y_i^* \right . \right )\right ]$,
 where $f(.,.)$ is a known function of the covariates and observed response.

\noindent To understand the dimensions of the variables we defined, let us consider the following simulated set of arbitrary data.\\
\begin{tabular} {  |p{0.5cm}||p{0.5cm}|p{0.5 cm}|p{0.5cm}||p{0.5cm}|p{0.5cm}|p{0.5cm}|p{0.5cm}||p{0.5cm}|p{0.5cm}|p{0.5cm}|p{1.5cm}|}
\hline
 $\mathbf{R}$ & $\mathbf{D}$ & $\mathbf{Y}$ & $\mathbf{Y}^\ast$ & $\mathbf{X}_1$ & $\mathbf{X}_2$ & $\hdots$ & $\mathbf{X}_p$ & $\bm{\beta}_0$ & $\bm{\beta}_1$ & $\hdots$ & $\bm{\beta}_p$\\
 \hline
 $R_1$  & $D_1$  & $Y_1$ & $Y_1^\ast$ & $X_{11}$ & $X_{12}$ & $\hdots$ & $X_{1p}$ & $\beta_{10}$ & $\beta_{11}$ & $\hdots$ & $\beta_{1(p-1)}$\\
\hline
 $R_2$  & $D_2$  & $Y_2$ & $Y_2^\ast$ & $X_{21}$ & $X_{22}$ & $\hdots$ & $X_{2p}$ & $\beta_{20}$ & $\beta_{21}$ & $\hdots$ & $\beta_{2(p-1)}$\\
\hline
 $\vdots$  & $\vdots$  & $\vdots$ & $\vdots$ & $\vdots$ & $\vdots$ & $\hdots$ & $\vdots$ & $\vdots$ & $\vdots$ & $\hdots$ & $\vdots$\\
\hline
 $\vdots$  & $\vdots$  & $\vdots$ & $\vdots$ & $\vdots$ & $\vdots$ & $\hdots$ & $\vdots$ & $\vdots$ & $\vdots$ & $\hdots$ & $\vdots$\\
\hline
 $\vdots$  & $\vdots$  & $\vdots$ & $\vdots$ & $\vdots$ & $\vdots$ & $\hdots$ & $\vdots$ & $\vdots$ & $\vdots$ & $\hdots$ & $\vdots$\\
\hline
 $R_n$  & $D_n$  & $Y_n$ & $Y_n^\ast$ & $X_{n1}$ & $X_{n2}$ & $\hdots$ & $X_{np}$ & $\beta_{n0}$ & $\beta_{n1}$ & $\hdots$ & $\beta_{n(p-1)}$\\
\hline
\end{tabular}
\\
\text{\\}
\textbf{Figure 1.}{ A visualized set of discrete data $\mathbf{Z}_t = (\mathbf{X}_t, Y_t, Y_t^\ast, D_t, R_t)$ and the parameters $\bm{\beta}$}
\\

\noindent Note that $Y_t$'s are the \textit{true} responses for $\mathbf{X}_t$ and $Y_t^\ast$'s are the deemed responses after the record linkage.
 Now consider the different outcomes of the Response variable when Review variable and Match variable are involved. For some $t$,\\
\begin{tabular} { |p{2.5cm}|p{2.5cm}|p{2.8cm}|p{2.8cm}|}
\hline
 Review variable & Match variable & Response variable & Predictor variable\\
 \hline
 $R=0$  & $--$  & $Y_t$ is unknown & $X_t$\\
\hline
$R=1$ & $D=0$ & $Y_t$ is unknown & $X_t$\\
\hline
$R=1$ & $D=1$ & $Y_t = Y_t^\ast$ & $X_t$\\
\hline
\end{tabular}


\section{Checking the score identity}
In the maximum-likelihood framework, under general regularity conditions, the Fisher information is equal to the variance-covariance matrix of the score function, see \cite{barndorff} pp.25. 
This key property leads to the asymptotic efficiency of maximum likelihood estimators through the Cramer-Rao bound.
In what follows, we show that this property is not satisfied by a score function, which is derived from the solution by Chipperfield et al. \cite{chipperfield_bishop_campbell_2011}.
This means that we can refine the related estimator to decrease its variance in large samples.
In this section, we assume a known conditional match probability $P \left (D_i \left | {\bf X}_i, Y_i^* \right . \right )$ for each value of the couple $\left ( {\bf X}_i, Y_i^*\right )$.
We also assume that for each function $f(.,.)$ of the covariates and observed response, the expectation
 $E \left [ f({\bf X}_i, Y_i^*) P \left (D_i \left | {\bf X}_i, Y_i^* \right . \right )\right ]$ is known.

The solution by Chipperfield et al. may be simply described as follows.
 First observe that with known responses the regression coefficients are the solution of the following classical equation.
\begin{equation}
\sum_{i=1}^n \mathbf{X}_i \left ( Y_i - \mu_i(\bm{\beta}) \right ) = 0 \label{logiteq}
\end{equation}
 Note that the above equation produces consistent estimators of the regression coefficients because unlinked records are assumed to occur completely at random.
 When some responses are not directly observed, Chambers et al. have suggested replacing each response by its conditional expectation given the observed data $O_i$, thus obtaining the following
 equation.
\begin{equation}
\label{eq: chipperfield logistic estimating equation}
\sum_{i=1}^n \mathbf{X}_i \left ( E \left [ \left . Y_i \right | O_i \right ] - \mu_i(\bm{\beta}) \right ) = 0
\end{equation}
 where the conditional expectation is computed as follows.
\begin{eqnarray}
\label{eq: chipperfield conditional mean response}
 E \left [ \left . Y_i \right | O_i \right ] &=&
 \underbrace{I(i \in s)}_{R_i} \left ( D_i Y_i^* + \left ( 1 - D_i \right ) \mu_i  \right ) + \nonumber \\
 & & \underbrace{I(i \notin s)}_{\left ( 1-R_i \right )} \left ( P(D_i=1|\mathbf{X}_i,Y_i^*) Y_i^* +  \mu_i P(D_i=0|\mathbf{X}_i,Y_i^*) \right )
\end{eqnarray}
 Note that this computation is exact only if the following condition is met.
\begin{equation}
\label{eq: implicit condition chipperfield}
E\left [ \left . Y_i \right | \mathbf{X}_i, Y_i^*, D_i=0 \right ] = E\left [ \left . Y_i \right | \mathbf{X}_i \right ] = \mu_i
\end{equation}
 This condition is different from the IAR assumption and is not implied by it.
 However exploring the exact nature of the relationship between these two assumptions is beyond the scope of this paper.
 The fact that the conditional expectation of the response $E \left [ \left . Y_i \right | O_i \right ]$ also depends on the regression coefficients is a challenge.
 To address this practical problem, Chipperfield et al. have proposed an iterative numerical solution that mimics the Expectation-Maximization (E-M) procedure \cite{em_dempster_1977}.
 This procedure includes an E-step and an M-step in each iteration.
 In the E-step, the current value of the parameter is used to compute the required conditional expectations.
 In the subsequent M-step, the parameter estimate is updated using the estimated conditional expectations.
 For the problem at hand, let $\bm{\beta}^{(t)}$ and $E \left [ \left . Y_i \right | O_i;\bm{\beta}^{(t)} \right ]$ denote respectively the estimated coefficients and
 conditional mean response for link $i$, in iteration $t$.
 Then $\bm{\beta}^{(t+1)}$ is computed in two steps as follows.
 First, compute $E \left [ \left . Y_i \right | O_i;\bm{\beta}^{(t)} \right ]$ using the current estimate $\bm{\beta}^{(t)}$
 and Equation~(\ref{eq: chipperfield conditional mean response}).
 Second, compute $\bm{\beta}^{(t+1)}$ as the solution of the following equation.
\begin{equation}
\sum_{i=1}^n \mathbf{X}_i \left ( E \left [ \left . Y_i \right | O_i;\bm{\beta}^{(t)} \right ]
 - \mu_i \left ( \bm{\beta}^{(t+1)} \right ) \right ) = 0
\end{equation}
 The solution to the above equation is found numerically using the iterative Newton-Raphson method or some variation of this procedure.
 The resulting numerical procedure is complex because it involves nested iterations within each M-step.

We next examine the properties of Equation~(\ref{eq: chipperfield logistic estimating equation}).
 First, note that this equation is of the following form.
\begin{equation}
\mathbf{S} = \sum_{i=1}^n \widetilde{\mathbf{S}}_i = 0
\end{equation}
 where
\begin{eqnarray}
 \widetilde{\mathbf{S}}_i &:=& \mathbf{X}_i (E[Y_i|O_i]-\mu_i \left ( \bm{\beta}\right )) \\
 &=&
\label{eq: score chipperfield}
 \mathbf{X}_i \left [ I(i\in s)D_i(Y_i^* - \mu_i)+ I(i\notin s)P(D_i|\mathbf{X}_i,Y_i^*)(Y_i^* - \mu_i) \right ]
\end{eqnarray}
 Equation~(\ref{eq: score chipperfield}) is a straightforward consequence of Equation~(\ref{eq: chipperfield conditional mean response}).
 Using this equation, we obtain the following equivalent form for Equation~(eq: chipperfield logistic estimating equation).
\begin{equation}
 \sum_{i=1}^n
 \mathbf{X}_i \left [ I(i\in s)D_i(Y_i^* - \mu_i \left ( \bm{\beta} \right ) ) +
 I(i\notin s)P(D_i|\mathbf{X}_i,Y_i^*)(Y_i^* - \mu_i \left ( \bm{\beta} \right )) \right ] = 0
\end{equation}
 This latter form suggests a simpler numerical procedure using the Newton-Raphson method, with no nested iterations.
 Now suppose that $\widetilde{\mathbf{S}}_i$ is the score of some proper likelihood, i.e. $\widetilde{\mathbf{S}}_i$ is based on the first-order
 partial derivatives of the corresponding log-likelihood.
 Then under regularity conditions, it must satisfy the following identity (see \cite{barndorff} pp.25), which is hereafter called {\it score identity}.
\begin{equation}
\label{eq: score identity}
 E\left[ - \frac{\partial \widetilde{\mathbf{S}}_i}{\partial\bm{\beta}^{\top}}\right] = E\left[\widetilde{\mathbf{S}}_i \widetilde{\mathbf{S}}_i^{\top}\right] 
\end{equation}
 With little loss of generality, we hereafter assume that all necessary regularity conditions hold.
 Our next goal is checking whether the score identity is always satisfied by $\widetilde{\mathbf{S}}_i$.
To this end, the following lemma is useful.
\begin{lem}
\label{lemma1}
	\begin{eqnarray}
		E\left[D_i(Y_i^\ast-\mu_i)^2|\mathbf{X}_i\right]=P(D_i=1|\mathbf{X}_i)\mu_i(1-\mu_i) \mbox{ for all } i
          \end{eqnarray}
\end{lem}

\begin{proof}
For all $i$,
\begin{eqnarray}
E\left[D_i(Y_i^\ast-\mu_i)^2|\mathbf{X}_i\right]&=& P(D_i=1|\mathbf{X}_i)E\left[D_i (Y_i^\ast-\mu_i)^2|\mathbf{X}_i, D_i=1 \right]+ \nonumber \\
                                                                                                                                                              & &  P(D_i=0|\mathbf{X}_i)E\left[D_i (Y_i^\ast-\mu_i)^2|\mathbf{X}_i,D_i=0\right] \\
&=& P(D_i=1|\mathbf{X}_i)E\left[ 1\cdot (Y_i^\ast-\mu_i)^2|\mathbf{X}_i, D_i=1 \right]+ \nonumber \\
                                                                                                                                                              & & P(D_i=0|\mathbf{X}_i)E\left[0\cdot (Y_i^\ast-\mu_i)^2|\mathbf{X}_i,D_i=0\right] \\
&=& P(D_i=1|\mathbf{X}_i)E\left[ (Y_i^\ast-\mu_i)^2|\mathbf{X}_i, D_i=1  \right]
\end{eqnarray}
\end{proof}

\noindent The next lemma shows that Equation~(\ref{eq: score identity}) is not always satisfied.
\begin{lem}\label{lemma: score identity difference}
\begin{eqnarray}
\label{eq: difference score}
E\left[ - \frac{\partial \widetilde{\mathbf{S}}_i}{\partial\bm{\beta}^{\top}}\right]- E\left[\widetilde{\mathbf{S}}_i \widetilde{\mathbf{S}}_i^{\top}\right] &=& (1-p)E\left[P(D_i=1|\mathbf{X}_i)\mu_i(1-\mu_i)\mathbf{X}_i \mathbf{X}_i^{\top}\right]- \nonumber \\
                                                                                                                   & & (1-p)E\left[P(D_i=1|\mathbf{X}_i,Y_i^\ast)^2 (Y_i^\ast-\mu_i)^2 \mathbf{X}_i \mathbf{X}_i^{\top}\right] \indent\indent
\end{eqnarray}
\end{lem}

\begin{proof}
Let us consider the right-hand side:
\begin{eqnarray}
\widetilde{\mathbf{S}}_i \widetilde{\mathbf{S}}_i^{\top} &=& \left[ I(i\in s)D_i(Y_i^\ast-\mu_i)+I(i\notin s)P(D_i = 1|\mathbf{X}_i,Y_i^\ast)(Y_i^\ast-\mu_i)\right]^2 \mathbf{X}_i \mathbf{X}_i^{\top}\indent \\
&=& \left[I(i\in s)D_i(Y_i^\ast-\mu_i)^2+I(i\notin s)P(D_i=1|\mathbf{X}_i,Y_i^\ast)^2(Y_i^\ast-\mu_i)^2 \right]\mathbf{X}_i \mathbf{X}_i^{\top}\indent \indent 
\end{eqnarray}

Hence \begin{eqnarray} 
E\left[\widetilde{\mathbf{S}}_i \widetilde{\mathbf{S}}_i^{\top}\right] &=& P(i\in s)E\left[D_i(Y_i^\ast-\mu_i)^2 \mathbf{X}_i \mathbf{X}_i^{\top}\right]+ \nonumber \\ 
																							& & P(i\notin s)E\left[P(D_i=1|\mathbf{X}_i,Y_i^\ast)^2 (Y_i^\ast-\mu_i)^2 \mathbf{X}_i \mathbf{X}_i^{\top}\right]
\end{eqnarray}

where 
\begin{eqnarray} E\left[D_i(Y_i^\ast-\mu_i)^2 \mathbf{X}_i \mathbf{X}_i^{\top}\right]&=& E\left[E\left[D_i(Y_i^\ast -\mu_i)^2 \mathbf{X}_i \mathbf{X}_i^{\top}|\mathbf{X}_i\right]\right] \nonumber \\ &=& E\left[E\left[D_i(Y_i^\ast-\mu_i)^2|\mathbf{X}_i\right]\mathbf{X}_i \mathbf{X}_i^{\top}\right] \nonumber  \\ &=& E\left[P(D_i=1|\mathbf{X}_i)\mu_i(1-\mu_i)\mathbf{X}_i \mathbf{X}_i^{\top}\right] \mbox{ Lemma~(\ref{lemma1}) } \nonumber \end{eqnarray}

Let $P(i\in s) = p$. Thus \begin{eqnarray} E\left[\widetilde{\mathbf{S}}_i \widetilde{\mathbf{S}}_i^{\top}\right] &=& pE\left[P(D_i=1|\mathbf{X}_i)\mu_i(1-\mu_i)\mathbf{X}_i \mathbf{X}_i^{\top}\right]+ \nonumber \\
															     & &(1-p)E\left[P(D_i=1|\mathbf{X}_i,Y_i^\ast)^2 (Y_i^\ast-\mu_i)^2 \mathbf{X}_i \mathbf{X}_i^{\top}\right] \end{eqnarray}

\noindent Now consider the left-hand side:
\begin{eqnarray} \frac{\partial\widetilde{\mathbf{S}}_i}{\partial\bm{\beta}^{\top}}&=&  \frac{\partial}{\partial\beta^{\top}}\mathbf{X}_i^{\top}\left[I(i\in s)D_i(Y_i^\ast-\mu_i)+I(i\notin s)P(D_i|\mathbf{X}_i,Y_i^\ast)(Y_i^\ast-\mu_i)\right] \\ 
&=& \frac{\partial}{\partial\bm{\beta}^{\top}}\mathbf{X}_i^{\top}\{-I(i\in s)D_i\mu_i-I(i\notin s)P(D_i|\mathbf{X}_i,Y_i^\ast)\mu_i\} \\
&=& \mathbf{X}_i^{\top}\{-I(i\in s)D_i-I(i\notin s)P(D_i|\mathbf{X}_i,Y_i^\ast)\}\frac{\partial}{\partial\bm{\beta}^{\top}}\mu_i \\
&=&  \mathbf{X}_i^{\top}\{-I(i\in s)D_i-I(i\notin s)P(D_i|\mathbf{X}_i,Y_i^\ast)\}\mu_i(1-\mu_i)\mathbf{X}_i \\
&=& \mu_i(1-\mu_i)\mathbf{X}_i\mathbf{X}_i^{\top}\{-I(i\in s)D_i-I(i\notin s)P(D_i|\mathbf{X}_i,Y_i^\ast)\} 
\end{eqnarray}

\noindent Then
\begin{eqnarray} E\left[\frac{\partial\widetilde{\mathbf{S}}_i}{\partial\bm{\beta}^{\top}}\right] &=& -E\left[\mu_i(1-\mu_i)\mathbf{X}_i\mathbf{X}_i^{\top}\{I(i\in s)D_i+I(i\notin s)P(D_i|\mathbf{X}_i,Y_i^\ast)\}\right]\\
&=& -\{E\left[\mu_i(1-\mu_i)\mathbf{X}_i\mathbf{X}_i^{\top}I(i\in s)D_i\right] + \nonumber \\ & & E\left[\mu_i(1-\mu_i)\mathbf{X}_i\mathbf{X}_i^{\top}I(i\notin s)P(D_i|\mathbf{X}_i,Y_i^\ast)\right]\}\\
&=& -\{P(i\in s)E\left[\mu_i(1-\mu_i)\mathbf{X}_i\mathbf{X}_i^{\top}D_i\right]+ \nonumber \\ & & P(i\notin s)E\left[\mu_i(1-\mu_i)\mathbf{X}_i\mathbf{X}_i^{\top}P(D_i|\mathbf{X}_i,Y_i^\ast)\right]\}\\
&=& -\{pE\left[\mu_i(1-\mu_i)\mathbf{X}_i\mathbf{X}_i^{\top}D_i\right]+ \nonumber \\ & & (1-p)E\left[\mu_i(1-\mu_i)\mathbf{X}_i\mathbf{X}_i^{\top}P(D_i|\mathbf{X}_i,Y_i^\ast)\right]\}\\
&=& -\{pE\left[E\left[D_i\mu_i(1-\mu_i)\mathbf{X}_i \mathbf{X}_i^{\top}|\mathbf{X}_i\right]\right]+ \nonumber \\ & & (1-p)E\left[E\left[P(D_i=1|\mathbf{X}_i,Y_i^\ast)\mu_i(1-\mu_i)\mathbf{X}_i \mathbf{X}_i^{\top}|\mathbf{X}_i\right]\right]\}\\
&=& -\{pE\left[P(D_i=1|\mathbf{X}_i)\mu_i(1-\mu_i)\mathbf{X}_i \mathbf{X}_i^{\top}\right]+ \nonumber \\ & & (1-p)E\left[P(D_i=1|\mathbf{X}_i)\mu_i(1-\mu_i)\mathbf{X}_i \mathbf{X}_i^{\top}\right]\}\\
&=& -E\left[P(D_i=1|\mathbf{X}_i)\mu_i(1-\mu_i)\mathbf{X}_i \mathbf{X}_i^{\top}\right] 
\end{eqnarray}

\noindent Hence 
\begin{eqnarray}
E\left[ - \frac{\partial\widetilde{\mathbf{S}}_i}{\partial\bm{\beta}^{\top}}\right]-E\left[\widetilde{\mathbf{S}}_i \widetilde{\mathbf{S}}_i^{\top}\right] 
&=& E\left[P(D_i=1|\mathbf{X}_i)\mu_i(1-\mu_i)\mathbf{X}_i \mathbf{X}_i^{\top}\right] - \nonumber \\    & &  pE\left[P(D_i=1|\mathbf{X}_i)\mu_i(1-\mu_i)\mathbf{X}_i \mathbf{X}_i^{\top}\right] - \nonumber \\ 
                                                                                                                                                                                                        & & (1-p)E\left[P(D_i=1|\mathbf{X}_i,Y_i^\ast)^2 (Y_i^\ast-\mu_i)^2 \mathbf{X}_i \mathbf{X}_i^{\top}\right]  \indent\indent\\
&=& (1-p)E\left[P(D_i=1|\mathbf{X}_i)\mu_i(1-\mu_i)\mathbf{X}_i \mathbf{X}_i^{\top}\right]- \nonumber \\
                                                                                                                   & & (1-p)E\left[P(D_i=1|\mathbf{X}_i,Y_i^\ast)^2 (Y_i^\ast-\mu_i)^2 \mathbf{X}_i \mathbf{X}_i^{\top}\right] \label{eq1} \indent\indent
\end{eqnarray}

\end{proof}

\noindent Next apply the above lemma under the following conditions, where the right-hand side of Equation~(\ref{eq: difference score}) is obviously positive definite.\\
\indent\indent(i) $P(D_i = 1|\mathbf{X}_i) = P(D_i=1|\mathbf{X}_i, Y_i^\ast) = \lambda \in (0,1)$\\
\indent\indent(ii) $\bm{\beta} = \left [ \beta_0, 0,\ldots, 0 \right ]^{\top}$, i.e. null slopes. Thus $Y_i$ is independent of $\mathbf{X}_i$ and
\begin{equation}
\mu_i = P(Y_i =1 |X_i) = P(Y_i=1) = \frac{e^{\beta_0}}{1 + e^{\beta_0}} = \varphi \in (0,1)
\end{equation}
\indent\indent(iii) $P(Y_i^\ast = 1) = \varphi$\\
\indent\indent (iv) $E\left[(Y_i^\ast-\mu_i)^2|\mathbf{X}_i\right] =  E\left[(Y_i^\ast-\varphi)^2|\mathbf{X}_i\right]= E\left[(Y_i^\ast-\varphi)^2\right] = \varphi (1-\varphi)$\\
\indent\indent (v) $E\left[\mathbf{X}_i\mathbf{X}_i^{\top}\right]$ is positive definite\\
 Then
\begin{eqnarray}
& & (1-p)\left\{E\left[P(D_i=1|\mathbf{X}_i)\mu_i(1-\mu_i)\mathbf{X}_i \mathbf{X}_i^{\top}\right]\right\} -  \nonumber \\  & & (1-p)\left\{E\left[P(D_i=1|\mathbf{X}_i,Y_i^\ast)^2 (Y_i^\ast-\mu_i)^2 \mathbf{X}_i \mathbf{X}_i^{\top}\right]\right\} \nonumber\\
&=& (1-p) \left\{ E\left[\lambda\varphi(1-\varphi)\mathbf{X}_i\mathbf{X}_i^{\top}\right]-E\left[\lambda^2 (Y_i^\ast-\varphi)^2\mathbf{X}_i\mathbf{X}_i^{\top}\right]\right\}  \\
&=&  (1-p) \left\{ E\left[\lambda\varphi(1-\varphi)\mathbf{X}_i\mathbf{X}_i^{\top}\right]-E\left[\lambda^2 E\left[(Y_i^\ast-\varphi)^2|\mathbf{X}_i\right]\mathbf{X}_i\mathbf{X}_i^{\top}\right]\right\}  \\
&=&  (1-p) \left\{ E\left[\lambda\varphi(1-\varphi)\mathbf{X}_i\mathbf{X}_i^{\top}\right]-E\left[\lambda^2 E\left[(Y_i^\ast-\varphi)^2\right]\mathbf{X}_i\mathbf{X}_i^{\top}\right]\right\} \\
&=&  (1-p) \left\{ E\left[\lambda\varphi(1-\varphi)\mathbf{X}_i\mathbf{X}_i^{\top}\right]-E\left[\lambda^2\varphi (1-\varphi)\mathbf{X}_i\mathbf{X}_i^{\top}\right]\right\} \\
&=&  (1-p) \left\{ \lambda\varphi(1-\varphi) E\left[\mathbf{X}_i\mathbf{X}_i^{\top}\right]-\lambda^2 \varphi (1-\varphi)E\left[\mathbf{X}_i\mathbf{X}_i^{\top}\right]\right\}\\
&=&  (1-p) \varphi (1-\varphi)(\lambda-\lambda^2)E\left[\mathbf{X}_i\mathbf{X}_i^{\top}\right] \\
&>& 0 
\end{eqnarray}


\section{The quasi-likelihood framework}

Lemma~\ref{lemma: score identity difference} suggests that it is possible to improve the estimating equation proposed
 by Chipperfield et al.\cite{chipperfield_bishop_campbell_2011}.
 To this end we need to apply the quasi-likelihood framework.
 This section provides some background on this framework based on the book by Heyde \cite{heyde}.

The quasi-likelihood provides a generalization of the maximum-likelihood framework in situations where the likelihood is intractable.
 It is also a unifying framework for maximum-likelihood estimation and least-square estimation.

\noindent Let $\{\mathbf{Z}_t, t\leq n\}$ be a size-$n$ sample of independent and identically distributed (iid) observations in $\mathbb{R}^m$.
 The distribution of $\mathbf{Z}_t$ depends on a "parameter"
 $\bm{\theta} = \left [ \theta_1 \ldots  \theta_p \right ]^{\top}$ taking values in an open subset $\bm{\Uptheta}$ of $\mathbb{R}^p$.
 The goal is estimating $\bm{\theta}_0$, the unknown true value of the parameter.
 In the quasi-likelihood framework, the parameter of interest is estimated by solving an equation of the following form.
\begin{equation}
\label{eq: general EE}
 \mathbf{G} \left ( \{\mathbf{Z}_t, t\leq n\}, \widehat{\bm{\theta}} \right ) = 0
\end{equation}
 where $\mathbf{G}(.,.)$ is an {\it estimating function} mapping $\mathbb{R}^{nm}  \times \mathbb{R}^p$ into $ \mathbb{R}^p$.
 For a given problem, there may be an infinite number of such estimating functions that form a class $\mathcal{H}$ of estimating functions.
 In this class, all estimating functions may produce consistent estimators with different precisions.
 In the quasi-likelihood framework, the goal is finding the most efficient estimating function, i.e. the one yielding the asymptotically most precise estimator within the class.
%
%
Following Heyde\cite{heyde}, we consider a class $\mathcal{H}$ of estimating functions where each member $\mathbf{G} = \left [G_1 \ldots G_p \right ]^{\top}$ is a mapping
 from $\mathbb{R}^{nm}  \times \mathbb{R}^p$ into $ \mathbb{R}^p$ and satisfies the following additional conditions.
\begin{enumerate}
\item Zero mean: $E \left [ \mathbf{G} \right ]=\mathbf{0}$ at $\bm{\theta}_0$\\
\item Square-integrable: $E\left [ \mathbf{G}\mathbf{G}^{\top} \right ] < \infty$\\
\item Positive definite variance-covariance matrix $E\left [ \mathbf{G}\mathbf{G}^{\top} \right ]$ at $\bm{\theta}_0$\\
\item Differentiable: continuously differentiable in a open neighborhood of $\bm{\theta}_0$ with derivative $\partial \mathbf{G}/\partial \bm{\theta}^{\top}$\\
\item Non-singular expected derivative at $E \left [ \partial \mathbf{G}/\partial \bm{\theta}^{\top} \right ]$ at $\bm{\theta}_0$
\end{enumerate}
 When it is convenient, the expected derivative at $E \left [ \partial \mathbf{G}/\partial \bm{\theta}^{\top} \right ]$ at $\bm{\theta}_0$ is also denoted by
 $E \left [ \left . \partial \mathbf{G}/\partial \bm{\theta}^{\top} \right |_{\bm{\theta}_0}\right ]$.
%
%
The matrix derivative $\partial \mathbf{G}/\partial \bm{\theta}^{\top}$ is defined according to Schott\cite{schott}, p. 327.
\begin{eqnarray}
 \frac{\partial\mathbf{G}}{\partial\bm{\theta}^{\top}} &=&
 \left [ \left ( \frac{\partial G_1}{\partial \bm{\theta}^{\top}} \right )^{\top} \ldots \left ( \frac{\partial G_p}{\partial \bm{\theta}^{\top}} \right )^{\top} \right ]^{\top} \\
 \frac{\partial G_r}{\partial \bm{\theta}^{\top}} &=& \left [ \frac{\partial G_r}{\partial \theta_1} \ldots \frac{\partial G_r}{\partial \theta_p} \right ], \ r=1,\ldots,p
\end{eqnarray}
%
%
 Under regularity condition, the estimator $\widehat{\bm{\theta}}$ that is based on Equation~(\ref{eq: general EE}) has the following asymptotic important properties \cite[pp. 40]{heyde}.
%
\begin{enumerate}
\item Consistency: $\widehat{\bm{\theta}} \stackrel{p}{\longrightarrow} \bm{\theta}_0$ as $n \longrightarrow \infty$,
 where $\stackrel{p}{\longrightarrow}$ denotes the convergence in probability\\
\item Asymptotic normal distribution: $\label{cond2} \widehat{\bm{\theta}} \dot\sim \mathcal{N}(\bm{\theta}_0, Var(\widehat{\bm{\theta}}))$\\
\end{enumerate}
 Furthermore, the estimator variance is approximated by the {\it sandwich formula} in large samples.
\begin{equation}
\label{cond3}
 Var(\widehat{\bm{\theta}})^{-1} \circeq \epsilon(\mathbf{G})
\end{equation}
 where $\circeq$ means "approximately equal" and $\epsilon(\mathbf{G})$ is the {\it Fisher information} that is given by the following expression.
\begin{equation}
\label{eq: fisher information}
 \epsilon(\mathbf{G}) = 
 E\left [ \left . \frac{\partial\mathbf{G}}{\partial\bm{\theta}^{\top}} \right |_{\bm{\theta}_0} \right ]^{\top}
 E\left [ \mathbf{GG}^{\top} \right ]^{-1}
 E\left [ \left . \frac{\partial\mathbf{G}}{\partial\bm{\theta}^{\top}} \right |_{\bm{\theta}_0} \right ]
\end{equation}
 Note that in the above equation, the right-hand side is evaluated at $\bm{\theta}_0$.
%
%
%
\noindent To minimize $Var(\widehat{\bm{\theta}})$, we need to maximize the Fisher information $\epsilon(\mathbf{G})$.
 Following Hedye \cite[Definition 2.1, pp. 12]{heyde}, call an estimating function $\mathbf{G}^*$ {\it $O_F$-optimal} within the class $\mathcal{H}$ if
 $\epsilon (\mathbf{G^\ast}) - \epsilon (\mathbf{G})$ is nonnegative definite for every member $\mathbf{G}$ of $\mathcal{H}$.
%
%
 For a given $\mathbf{G} \in \mathcal{H}$, it is also convenient to define the standardized estimating function as follows.
%
\begin{equation} \mathbf{G}^{(s)} =
 - E\left[ \frac{\partial\mathbf{G}}{\partial\bm{\theta}^{\top}}\right]^{\top}  \left(E\left[\mathbf{G} \mathbf{G}^{\top} \right]\right)^{-1} \mathbf{G}
\end{equation}

\noindent The following key theorem gives two equivalent sufficient conditions for $O_F$-optimality within a class of estimating functions.
%
\begin{thm}
\label{imptthm}
\cite[Theorem 2.1, pp. 14]{heyde} $\mathbf{G^\ast}\in\mathcal{H}$ is an $O_F$-optimal estimating function within $\mathcal{H}$ if 
\begin{equation}
 E\left[ \mathbf{G}^{\ast (s)} \mathbf{G}^{(s)^{\top}}\right]  = E\left [ \mathbf{G}^{(s)} \mathbf{G}^{{\ast (s)}^{\top}} \right ] =
 E\left [ \mathbf{G}^{(s)} \mathbf{G}^{(s)^{\top}} \right ]
\end{equation}
%
%
or equivalently
\begin{equation}
\left(E\left[\frac{\partial\mathbf{G}}{\partial\bm{\theta}^{\top}}\right]\right)^{-1} E\left[\mathbf{G} {\mathbf{G^\ast}}^{\top}\right]
\end{equation}
 is a constant matrix $\forall\mathbf{G}\in\mathcal{H}$.

\end{thm}


\noindent We next apply the above theorem to the following general class of estimating functions.
 Let $\left ( \mathbf{X}_i, \mathbf{Y}_i \right )$ be a size-$n$ iid sample, with a distribution that depends on an unknown vector of parameters $\bm{\theta}$, such that for
 the true parameter $\bm{\theta}_0$, we have
\begin{equation}
\label{eq: conditional moment condition}
 E \left [ \left . \mathbf{h} \left ( \mathbf{X}_i, \mathbf{Y}_i, \bm{\theta}_0 \right )\right | \mathbf{X}_i = \mathbf{x}_i \right ] = 0, \ \forall \mathbf{x}_i
\end{equation}
 where  $\mathbf{h} \left (.,.,. \right )$ is a function into $\mathbb{R}^d$.
 For convenience define $ \mathbf{H}_i \left (\bm{\theta} \right ) =  \mathbf{h} \left ( \mathbf{X}_i, \mathbf{Y}_i, \bm{\theta} \right )$ that is simply denoted
 by $\mathbf{H}_i$ if it is clear from the context.
 The notation defined in previous paragraphs applies with $\mathbf{Z}_i$ comprising of $\mathbf{X}_i$ and $\mathbf{Y}_i$.
 This setup covers the situation where $\mathbf{X}_i$ is a vector of covariates and $\mathbf{Y}_i$ is a vector of responses, including when we have a single scalar response.
 However, it also covers more general situations when the vector $\mathbf{Y}_i$ includes the responses and other covariates.
 Let $\mathbf{X} = \left [ \mathbf{X}_1 \ldots \mathbf{X}_n \right ]^{\top}$ and $\mathcal{X}$ denote the design matrix and $\mathcal{X}$ the set of such matrices,
 that have $n$ rows, the appropriate number of columns and a first column of ones.
 Also for further convenience define the $q$-dimensional column vector $\mathbf{H} = \left [ \mathbf{H}_1^{\top} \ldots \mathbf{H}_n^{\top} \right ]^{\top}$ where $q=nd$.
 Note that $\mathbf{H}_i$ is {\it not} the $i$-th component of the function $\mathbf{h}(.,.,.)$.
 Equation~(\ref{eq: conditional moment condition}) and the iid nature of the observations imply that
\begin{equation}
 E \left [ \left . \mathbf{H} \right | \mathbf{X}=\mathbf{x} \right ] = \mathbf{0}, \ \forall \mathbf{x} \in \mathcal{X}
\end{equation}
 Let $\mathcal{A}$ denote the set of functions that map $\mathcal{X}$ to a $p \times q$ matrix, such that $\mathbf{AH}$ satisfies
 the square-integrability, differentiability, positive-definite variance-covariance at $\bm{\theta}_0$, and
 non-singular expected derivative at $\bm{\theta}_0$.
%
%
 We are interested in estimating functions in the following class.
\begin{equation}
\label{eq: conditional class of EFs}
\mathcal{H} = \left\{\mathbf{G}: \mathbf{G}=\mathbf{A}(\mathbf{X}) \mathbf{H}(\bm{\theta}) \ s.t. \ \mathbf{A} \in \mathcal{A} \right\}
\end{equation}
 Note that the above definitions of $\mathcal{A}$ and $\mathcal{H}$ imply that
 $E \left [ \left . \partial \mathbf{G} /\partial \bm{\theta}^{\top}\right |_{\bm{\theta}_0} \right ]$ is nonsigular for any $\mathbf{G} \in \mathcal{H}$.
 For this class of estimating functions, consider the following choice.
\begin{equation}
 \mathbf{A}^* =
 \left(E\left[ \left . \frac{\partial\mathbf{H}}{\partial\bm{\theta}^{\top}} \right | \mathbf{X} \right]\right)^{\top}
 \left(E\left[ \left . \mathbf{H} \mathbf{H}^{\top} \right | \mathbf{X} \right]\right)^{-1}
\end{equation}
 where $\mathbf{H}$ and its derivatives are evaluated at $\bm{\theta}_0$ inside the different expectations.
 This choice satisfies the second form of the sufficient condition in Theorem \ref{imptthm}. Indeed
\noindent Then
\begin{eqnarray}
E\left[ \mathbf{G} \mathbf{G}^{* \top}  \right] &=& E\left[ \left( \mathbf{AH} \right)\left(\mathbf{A}^* \mathbf{H} \right)^{\top} \right]\\
&=& E\left[ \left( \mathbf{AH} \right)\left(\mathbf{H}^{\top} \mathbf{A}^{* {\top}} \right) \right]\\
&=& E\left[ E\left[ \left . \mathbf{AH}\mathbf{H}^{\top} \mathbf{A}^{* {\top}} \right | \mathbf{X} \right] \right]\\
&=& E\left[ \mathbf{A} E\left[ \left . \mathbf{H} \mathbf{H}^{\top} \right | \mathbf{X} \right] \mathbf{A}^{* {\top}} \right] \\
&=&
 E\left[ \mathbf{A} E\left[ \left . \mathbf{H}\mathbf{H}^{\top} \right | \mathbf{X}\right]
 \left\{   E\left[ \mathbf{H}\mathbf{H}^{\top} \rvert \mathbf{X} \right]^{-1}
 E\left[ \left . \frac{\partial \mathbf{H}}{\partial\bm{\theta}^{\top}} \right | \mathbf{X} \right] \right\}   \right] \\
&=& E\left[ \mathbf{A} E\left[ \left . \frac{\partial \mathbf{H}}{\partial\bm{\theta}^{\top}} \right | \mathbf{X} \right]   \right] \\
&=& E\left[ E\left[ \left . \frac{\partial \mathbf{A} \mathbf{H}}{\partial\bm{\theta}^{\top}} \right | \mathbf{X} \right]   \right] \\
&=& E\left[\frac{\partial\mathbf{G}}{\partial\bm{\theta}^{\top}} \right]
\end{eqnarray}
 The above equation implies the sufficient condition because $E \left [ \left . \partial \mathbf{G} /\partial \bm{\theta}^{\top}\right |_{\bm{\theta}_0} \right ]$ is nonsigular
 by assumption.
%
%
 Thus, an $O_F$-optimal estimating function within $\mathcal{H}$ is as follows.
\begin{eqnarray}
 \mathbf{G}^* &=&
 E\left[ \left . \frac{\partial\mathbf{H}}{\partial\bm{\theta}^{\top}} \right | \mathbf{X} \right]^{\top}
 \left(E\left[ \left. \mathbf{H} \mathbf{H}^{\top} \right | \mathbf{X} \right]\right)^{-1} \mathbf{H} \\
 &=&
\label{eq: optimal EE conditional}
 \sum_{i=1}^{n} E\left[ \left . \frac{\partial\mathbf{H}_i}{\partial\bm{\theta}^{\top}} \right | \mathbf{X}_i \right]^{\top}
 \left(E\left[ \left . \mathbf{H}_i \mathbf{H}_i^{\top} \right | \mathbf{X}_i \right]\right)^{-1} \mathbf{H}_i
\end{eqnarray}




\text{\\}


\section{Construction of an optimal estimating equation}

\noindent In this section, we build an estimator with a smaller asymtotic variance than the one proposed by Chipperfield et al..
 This estimater is based on the quasi-likehood framework.
 To this end, we apply the results of the previous section with $\mathbf{Z}_t = \left ( \mathbf{X}_t, Y_t, Y_t^*, R_t, D_t \right ) \in \mathbb{R}^{p+4}$, for $t=1,\ldots,n$,
 where $\mathbf{X}_t, Y_t, Y_t^*$, and $D_t$ are defined as in Section~\ref{section: notations}.
 Define the function $h()$ as follows.
\begin{equation}
\label{eq: base logistic EF}
 h\left ( \mathbf{X}_t, Y_t, Y_t^*, R_t, D_t, \bm{\beta} \right ) =
\left\{R_t D_t + (1-R_t) P(D_t=1|\mathbf{X}_t, Y_t^*)\right\}\left ( Y_t^*-\mu \left ( \bm{\beta};\mathbf{X}_t \right ) \right )
\end{equation}
 Let $H_t \left ( \bm{\beta} \right )$ denote $h\left ( \mathbf{X}_t, Y_t, Y_t^*, R_t, D_t, \bm{\beta} \right )$.
 and define $\mathbf{H}$ as in the previous section.
 The following lemma shows that $E \left [ \left . H_t \right | \mathbf{X}_t \right ]=0$ hence $E \left [ \left . \mathbf{H} \right | \mathbf{X} \right ]=0$.
\begin{lem}
$E\left[\{I(i\in s)D_i + I(i\notin s)P(D_i = 1|\mathbf{X}_i, Y_i^\ast)\}(Y_i^\ast-\mu_i)|\mathbf{X}_i\right] = 0$
\begin{proof}
(1) First we want to show that $E\left[I(i\in s)D_i(Y_i^\ast=\mu_i)|\mathbf{X}_i\right]=0$ \\
\\Consider 
\begin{eqnarray}
E\left[I(i\in s)D_i(Y_i^\ast-\mu_i)|\mathbf{X}_i\right] &=& P(i\in s)E\left[D_i(Y_i^\ast-\mu_i)|\mathbf{X}_i\right]\\ 
&=& P(i\in s)E\left[E\left[D_i(Y_i^\ast - \mu_i)|\mathbf{X}_i, D_i\right]|\mathbf{X}_i\right]\\ 
&=& P(i\in s)E\left[D_iE\left[(Y_i^\ast - \mu_i)|\mathbf{X}_i, D_i\right]|\mathbf{X}_i\right] \\
&=& 0
\end{eqnarray}
since (15) is implied by the definition of non-informative linkage,\\
\begin{eqnarray}
D_i E\left[(Y_i^\ast - \mu_i)|\mathbf{X}_i, D_i\right]
&=&\left\{\begin{array}{cl} 0&\mbox{if $D_i = 0$}\\ E[Y_i^\ast - \mu_i|X_i,D_i=1]& \mbox{if $D_i = 1$}\end{array}\right. \\
&=&\left\{\begin{array}{cl} 0&\mbox{if $D_i = 0$}\\ E[Y_i - \mu_i|X_i,D_i=1]& \mbox{if $D_i = 1$}\end{array}\right. \\
&=&\left\{\begin{array}{cl} 0&\mbox{if $D_i = 0$}\\ E[Y_i - \mu_i|X_i]=0& \mbox{if $D_i = 1$}\end{array}\right. \\
&=&\left\{\begin{array}{cl} 0&\mbox{if $D_i = 0$}\\ 0& \mbox{if $D_i = 1$}\end{array} \right. 
\end{eqnarray}
\\
\indent\indent (2) Second we want to show that $$E\left[I(i\notin s)P(D_i=1|\mathbf{X}_i, Y_i^\ast)(Y_i^\ast-\mu_i)|\mathbf{X}_i\right]=0$$
\indent\indent\indent\hspace{1mm} Consider 
\begin{eqnarray}
& & P(i\notin s)E\left[P(D_i=1|\mathbf{X}_i,Y_i^\ast)(Y_i^\ast-\mu_i)|\mathbf{X}_i\right]\\
&=& P(i\notin s)E\left[P(D_i=1|\mathbf{X}_i, Y_i^\ast)(Y_i-\mu_i)|\mathbf{X}_i\right]\\
&=&P(i\notin s)E\left[E\left[D_i(Y_i^\ast-\mu_i)|\mathbf{X}_i,Y_i^\ast\right]|\mathbf{X}_i\right]\\
&=&P(i\notin s)E\left[D_i(Y_i^\ast-\mu_i)|\mathbf{X}_i\right] \indent \mbox{ by the tower property }\\
&=&0 
\end{eqnarray}
\begin{adjustwidth}{1cm}{}
Lastly, we can apply the linearity of expectation meaning that the conditional expectation of two summands is the sum of the two conditional expectations (both 1 and 2), and thus the lemma holds.
\end{adjustwidth}
\end{proof}
\end{lem}
 The lemma shows that $\mathbf{H}$ satisfies Equation~(\ref{eq: conditional moment condition}).
 Note that, the lemma applies even if Equation~(\ref{eq: implicit condition chipperfield}) does not hold.
 We are interested in the class of estimating functions given by Equation~(\ref{eq: conditional class of EFs}).
 The corresponding optimal estimating function is given by Equation~(\ref{eq: optimal EE conditional}), which is rewritten as follows because $H_i$ is a scalar for
 our specific problem.
\begin{eqnarray}
 \mathbf{G}^* &=&
 \sum_{i=1}^{n} E\left[ \left . \frac{\partial {H}_i}{\partial\bm{\theta}^{\top}} \right | \mathbf{X}_i \right]^{\top}
 \left(E\left[ \left . {H}_i {H}_i^{\top} \right | \mathbf{X}_i \right]\right)^{-1} {H}_i \\
 &=&
 \sum_{i=1}^{n} \mathbf{A}_i^* {H}_i
\end{eqnarray}
 where
\begin{equation}
 \mathbf{A}_i^* =
 E\left[ \left . \frac{\partial {H}_i}{\partial\bm{\theta}^{\top}} \right | \mathbf{X}_i \right]^{\top}
 \left(E\left[ \left . {H}_i {H}_i^{\top} \right | \mathbf{X}_i \right]\right)^{-1}
\end{equation}
 The next lemma computes the right-hand side of the above equation.
%
%
%
%
%
\begin{lem}
\begin{equation}
\label{eq: optimal multiplier matrix logistic RL}
\mathbf{A}_i^* =
-\mathbf{X}_i
\frac{\mu_i (1-\mu_i) P\left ( D_i=1 |\mathbf{X}_i\right )}{p \mu_i (1-\mu_i) P(D_i=1|\mathbf{X}_i)+
(1-p) E\left[\left. P(D_i=1|\mathbf{X}_i,Y_i^\ast)^2 (Y_i^\ast-\mu_i)^2 \right | \mathbf{X}_i\right]}
\end{equation}

\begin{proof}
We have
\begin{eqnarray}
& & E\left[\frac{\partial H_i}{\partial\beta}|\mathbf{X}_i\right] \nonumber \\
&=& -\{E\left[I(i\in s)D_i|\mathbf{X}_i\right]+E\left[I(i\notin s)P(D_i=1|\mathbf{X}_i,Y_i^\ast)|\mathbf{X}_i\right]\}\frac{\partial\mu_i}{\partial\beta}\\
&=& -\{p\cdot E\left[D_i|\mathbf{X}_i\right]+(1-p)E\left[E\left[D_i|\mathbf{X}_i,Y_i^\ast\right]|\mathbf{X}_i\right]\}\frac{\partial\mu_i}{\partial\beta}\\
&=& -\{p\cdot E\left[D_i|\mathbf{X}_i\right]+(1-p)E\left[D_i|\mathbf{X}_i\right]\}\frac{\partial\mu_i}{\partial\beta} \mbox { by Law of Total Expectation }\indent\\
&=& -E\left[D_i|\mathbf{X}_i\right]\frac{\partial\mu_i}{\partial\beta}\\
&=& -E\left[D_i|\mathbf{X}_i\right]\mu_i(1-\mu_i)\mathbf{\mathbf{X}_i}
\end{eqnarray}

\noindent For $Var(H_i|\mathbf{X}_i)$, first note that
\begin{equation}
Var(H_i|\mathbf{X}_i) = E\left[Var(H_i|\mathbf{X}_i,I(i\in s))|\mathbf{X}_i\right]+Var(E\left[H_i|\mathbf{X}_i,I(i\in s)\right]|\mathbf{X}_i)
\end{equation}
 Then
\begin{eqnarray} 
& & E\left[H_i|\mathbf{X}_i,I(i\in s)\right]\\
&=&  E\left[I(i\in s)D_i(Y_i^\ast-\mu_i)|\mathbf{X}_i,I(i\in s)\right]+ \nonumber \\ & & E\left[I(i\notin s)P(D_i=1|\mathbf{X}_i,Y_i^\ast)(Y_i^\ast-\mu_i)|\mathbf{X}_i,I(i\in s)\right]\\
&=& I(i\in s)E\left[D_i(Y_i^\ast-\mu_i)|\mathbf{X}_i,I(i\in s)\right]+ \nonumber \\ & & I(i\notin s)E\left[P(D_i=1|\mathbf{X}_i,Y_i^\ast)(Y_i^\ast-\mu_i)|\mathbf{X}_i,I(i\in s)\right]\\
&=& I(i\in s)E\left[D_i(Y_i^\ast-\mu_i)|\mathbf{X}_i,I(i\in s)\right]+ \nonumber \\ & & I(i\notin s)E\left[E\left[D_i (Y_i^\ast-\mu_i)|\mathbf{X}_i,Y_i^\ast\right]|\mathbf{X}_i,I(i\in s)\right] \\
&=& I(i\in s)E\left[D_i(Y_i^\ast-\mu_i)|\mathbf{X}_i\right]+I(i\notin s)E\left[D_i (Y_i^\ast-\mu_i)|\mathbf{X}_i,I(i\in s)\right]\\
&=& I(i\in s)\cdot 0+I(i\notin s)\cdot 0\\
&=& 0
\end{eqnarray}

\noindent Then 
\begin{eqnarray}
& & Var(H_i|\mathbf{X}_i) \\
&=& E\left[Var(H_i|\mathbf{X}_i,I(i\in s)|\mathbf{X}_i\right]+Var(0|\mathbf{X}_i)\mbox{ by (38) }\\
 &=& E\left[Var(H_i|\mathbf{X}_i,I(i\in s))|\mathbf{X}_i\right]\\
 &=& P(i\in s|\mathbf{X}_i)Var(H_i|\mathbf{X}_i, i\in s)+P(i\notin s|\mathbf{X}_i)(Var(H_i|\mathbf{X}_i, i\notin s) \indent\indent\\
 &=& p\cdot Var(H_i|\mathbf{X}_i, i\in s)+(1-p)Var(H_i|\mathbf{X}_i, i\notin s)
\end{eqnarray}

\noindent Now consider 
\begin{eqnarray}
& & Var(H_i|\mathbf{X}_i, i\in s) \\
&=& Var(D_i(Y_i^\ast-\mu_i)|\mathbf{X}_i, i\in s) \\ 
&=& Var(D_i(Y_i^\ast-\mu_i)|\mathbf{X}_i) \\
&=& E\left[Var(D_i(Y_i^\ast-\mu_i)|\mathbf{X}_i,D_i)|\mathbf{X}_i\right]+Var(E\left[D_i(Y_i^\ast-\mu_i)|\mathbf{X}_i,D_i\right]|\mathbf{X}_i) \indent\\
&=& E\left[Var(D_i(Y_i^\ast-\mu_i)|\mathbf{X}_i,D_i)|\mathbf{X}_i\right]+Var(0|\mathbf{X}_i) \\
&=& E\left[Var(D_i(Y_i^\ast-\mu_i)|\mathbf{X}_i,D_i)|\mathbf{X}_i\right]+ 0 \\
&=& E\left[Var(D_i(Y_i^\ast-\mu_i)|\mathbf{X}_i,D_i)|\mathbf{X}_i\right] \\
&=& E\left[Var(Y_i^\ast-\mu_i|\mathbf{X}_i,D_i=1|\mathbf{X}_i\right] \\
&=& E\left[D_i \cdot Var(Y_i-\mu_i|\mathbf{X}_i)|\mathbf{X}_i\right] \\
&=& E\left[D_i \cdot \mu_i (1-\mu_i)|\mathbf{X}_i\right] \\
&=& \mu_i (1-\mu_i)\cdot E\left[D_i|\mathbf{X}_i\right] \\
&=& \mu_i (1-\mu_i)\cdot E\left[E\left[D_i|\mathbf{X}_i,Y_i^\ast\right]|\mathbf{X}_i\right] \\
&=&\mu_i (1-\mu_i)\cdot E\left[P(D_i=1|\mathbf{X}_i,Y_i^\ast)|\mathbf{X}_i\right] \\
&=& \mu_i (1-\mu_i)\cdot P(D_i=1|\mathbf{X}_i)
\end{eqnarray}

\noindent Also consider 
\begin{eqnarray}
& & Var(H_i|\mathbf{X}_i, i\notin s) \indent \indent \indent \indent\indent \indent \indent\indent  \indent \indent  \\
&=& Var(P(D_i=1|\mathbf{X}_i,Y_i^\ast)(Y_i^\ast-\mu_i)|\mathbf{X}_i)\indent \indent \indent \indent \indent \indent \indent  \indent\indent \indent  \\ 
&=& E\left[\{P(D_i=1|\mathbf{X}_i,Y_i^\ast)(Y_i^\ast-\mu_i)\}^2 | \mathbf{X}_i\right]+ \nonumber \\ & & (E\left[P(D_i=1|\mathbf{X}_i,Y_i^\ast)(Y_i^\ast-\mu_i)|\mathbf{X}_i\right])^2 \indent \indent \indent \indent\indent \indent \indent\indent \indent\indent   \\ 
&=& E\left[\{P(D_i=1|\mathbf{X}_i,Y_i^\ast)(Y_i^\ast-\mu_i)\}^2 | \mathbf{X}_i\right]+ (0)^2\indent \indent \indent \indent \indent \indent \indent\indent\indent\indent    \\ 
&=&  E\left[P(D_i=1|\mathbf{X}_i,Y_i^\ast)^2 (Y_i^\ast-\mu_i)^2 | \mathbf{X}_i\right]\indent \indent \indent \indent \indent \indent \indent \indent \indent \indent 
\end{eqnarray}

\noindent Thus, 
\begin{eqnarray}
& & Var(H_i|\mathbf{X}_i) \nonumber \indent \indent \indent \indent \\ 
&=& p\cdot\mu_i (1-\mu_i)P(D_i=1|\mathbf{X}_i,Y_i^\ast)+ \nonumber \\ & & (1-p)\cdot E\left[P(D_i=1|\mathbf{X}_i,Y_i^\ast)^2 (Y_i^\ast-\mu_i)^2 | \mathbf{X}_i\right] \indent\indent \indent\indent \indent\indent \indent\indent \indent\indent 
\end{eqnarray}

\noindent Thus Equation~(\ref{eq: optimal multiplier matrix logistic RL}) holds.

\end{proof}
\end{lem}

\text{\\}


\section{Proposed estimation procedure}

\noindent We aim to estimate $\bm{\theta}_0$ by the solution $\widetilde{\bm{\theta}}$ of the following estimating equation
\begin{equation}
 \sum_{i=1}^{n} \mathbf{A}_i^* {H}_i = 0
\end{equation}
 where $H_i$ is given by Equation~(\ref{eq: base logistic EF}) and
 $\mathbf{A}_i^*$ is given by Equation~(\ref{eq: optimal multiplier matrix logistic RL}).
 There are two obstacles to this implementation.
 First the optimal multiplier matrix $\mathbf{A}_i^*$ depends on $\bm{\theta}_0$ and on the conditional match probability $P(D_i=1|\mathbf{X}_i,Y_i^\ast)$,
 which must be estimated from the clerical sample.
 Second $H_i$ also depends on the conditional match probability $P(D_i=1|\mathbf{X}_i,Y_i^\ast)$.
 A solution is to use a two-step estimator as follows.
 In a preliminary step, estimate $P(D_i=1|\mathbf{X}_i,Y_i^\ast)$ from the clerical-sample, ideally for each observable $(\mathbf{X}_i,Y_i^\ast)$ pair.
 The following ratio estimator may be used.
\begin{equation}
\widehat{P}(D_i=1|\mathbf{X}_i,Y_i^\ast) = \frac{\#\{D_i=1|\mathbf{X}_i, Y_i^\ast\}}{\#\{D_i=1|\mathbf{X}_i, Y_i^\ast\} + \#\{D_i=0|\mathbf{X}_i, Y_i^\ast\}}
\end{equation}
 Next, use this estimate in the estimating equation by Chipperfield et al. to obtain a first-step estimate $\widetilde{\bm{\theta}}^{(1)}$.
 Also estimate $E\left[P(D_i=1|\mathbf{X}_i,Y_i^\ast)^2 (Y_i^\ast-\mu_i)^2 | \mathbf{X}_i\right]$ using the first-step estimate and the estimated
 conditional match probability.
 For example the following ratio estimator may be used, where $\widehat{\mu}_i$ is based on the first-step estimator.
\begin{equation}
\widehat{E}\left[P(D_i=1|\mathbf{X}_i,Y_i^\ast)^2 (Y_i^\ast-\mu_i)^2|\mathbf{X}_i\right] =
\frac{\sum_{j=1}^n I(\mathbf{X}_j=\mathbf{X}_i)\{P(D_j=1|\mathbf{X}_j,Y_j^\ast)^2(Y_j^\ast-\widehat{\mu}_j)^2\}}{\sum_{j=1}^nI(\mathbf{X}_j=\mathbf{X}_i)}
\end{equation}
 Next, use the first step estimate to compute an estimate $\widehat{\mathbf{A}}_i^*$ and compute the second step estimate $\widetilde{\bm{\theta}}^{(2)}$
 as the solution of the following estimating equation.
\begin{equation}
 \sum_{i=1}^{n} \widehat{\mathbf{A}}_i^* {H}_i = 0
\end{equation}


\section{Conclusion and future work}
In future work, we will implement the proposed estimator and compare it to the original estimator by Chipperfield et al. \cite{chipperfield_bishop_campbell_2011},
 to assess the gain in efficiency.

\section*{Acknowledgments}  It is a pleasure to thank Professor Michael Evans for his valuable advice.

\end{document}